\documentclass[11pt,a4paper,oneside]{article}
\usepackage[top=3cm, bottom=3cm, left=2cm, right=2cm]{geometry}
\usepackage[english]{babel}
\usepackage[utf8x]{inputenc}
\usepackage[T1]{fontenc}
\usepackage{enumitem}
\usepackage{amsthm,amsmath,amssymb,mathrsfs,dsfont,mathtools}

\usepackage{bm}
\usepackage{graphicx}
\usepackage[colorinlistoftodos]{todonotes}
\usepackage[colorlinks=true, allcolors=blue]{hyperref}
\usepackage{booktabs,subcaption}
\usepackage{hyperref}
\usepackage[all]{nowidow}

\usepackage{setspace}
\usepackage{algorithm}
\usepackage[noend]{algpseudocode}
\usepackage{tikz} 
\usetikzlibrary{shapes.geometric, arrows}
\tikzstyle{process} = [rectangle, rounded corners, minimum width=4cm, minimum height=1.2cm, text centered, draw=black]
\tikzstyle{arrow} = [thick,->,>=stealth]

\DeclarePairedDelimiter\abs{\lvert}{\rvert}%
\DeclarePairedDelimiter\norm{\lVert}{\rVert}%

\makeatletter
\let\oldabs\abs
\def\abs{\@ifstar{\oldabs}{\oldabs*}}
\let\oldnorm\norm
\def\norm{\@ifstar{\oldnorm}{\oldnorm*}}
\makeatother

\usepackage[sf]{titlesec}
\usepackage{sectsty}
\allsectionsfont{\centering \normalfont\scshape} 

\usepackage[round]{natbib}

\usepackage{authblk}
\title{Bayesian nonparametric modelling of sequential discoveries}
\author[1]{Alessandro Zito}
\author[2]{Tommaso Rigon}
\author[3]{Otso Ovaskainen}
\author[1]{David Dunson}
\affil[1]{Department of Statistical Science, Duke University, Durham, North 
Carolina 27708, U.S.A.}
\affil[2]{Department of Economics, Management and Statistics, University of Milano--Bicocca, 20126 Milano, Italy}
\affil[3]{Organismal and Evolutionary Biology Research Programme, University of Helsinki, Helsinki, Finland}

\date{}
\newtheorem{theorem}{Theorem}
\newtheorem{corollary}{Corollary}

\newtheorem{proposition}{Proposition}

\theoremstyle{definition}
\newtheorem{definition}{Definition}
\newtheorem{remark}{Remark}

\makeatletter
\def\env@cases{%
  \let\@ifnextchar\new@ifnextchar
  \left\lbrace
  \def\arraystretch{0.8}%
  \array{@{}l@{\quad}l@{}}}
\makeatother

\usepackage{subfiles}

\begin{document}
\maketitle

\begin{abstract}
We aim at modelling the appearance of distinct tags in a sequence of labelled objects. Common examples of this type of data include words in a corpus or distinct species in a sample. These sequential discoveries are often summarised via accumulation curves, which count the number of distinct entities observed in an increasingly large set of objects. We propose a novel Bayesian nonparametric method for species sampling modelling by directly specifying the probability of a new discovery, therefore allowing for  flexible specifications. The asymptotic behavior and finite sample properties of such an approach are extensively studied. Interestingly, our enlarged class of sequential processes includes highly tractable special cases. We present a subclass of models characterized by appealing theoretical and computational properties. Moreover, due to strong connections with logistic regression models, the latter subclass can naturally account for covariates. We finally test our proposal on both synthetic and real data, with special emphasis on a large fungal biodiversity study in Finland. 
\end{abstract}

\section{Introduction}

Our goal is to develop a flexible procedure for modelling the appearance of previously unobserved objects in a sequence. The sequential recording of distinct entities can be represented through an \textit{accumulation curve}, namely the cumulative number of distinct entities $K_n$ within a collection of $n$ objects. These entities can be of various nature, including biological species \citep{Good_1953, Good_Toulmin_1956}, words \citep{Efron_tristed_1976, Tristed_efron_198}, genes \citep{Ionita-Laza5008}, bacteria \citep{Hughes_2001, Gao_2007} and cell types \citep{Camerlenghi_2019}. 
The analysis of accumulation curves has a rich history in statistics, as testified by the early contributions of \citet{fisher_1943},  \citet{Good_1953},  and \citet{Good_Toulmin_1956}. We refer  to \citet{Bunge_Fitzpatrick_1993} for a historical account on the topic.  Several nonparametric approaches have been developed in more recent years, often aiming at i) predicting the number of unseen entities \citep[e.g. ][]{Shen_Chao_Lin_2003}, 
or ii) estimating the probability of a new discovery \citep[e.g. ][]{Chao_Shen_2004, Mao_2004, Favaro_2012}.

Our work builds on Bayesian nonparametric methods, whose development has been spurred by the seminal paper of \citet{Ferguson1973} on the Dirichlet process. In our motivating application, we aim to assess how many of the species present in a sample are missed when a given number of \textsc{dna} barcode sequences are obtained through high-throughput sequencing.
 Let $(X_n)_{n \geq 1}$ be a sequence of objects, such as fungal \textsc{dna} sequences in a single biological soil or air sample \citep{Abrego_2020}, taking values in $\mathds{X}$, which is the space of fungal species in our case. Among the first $n$ observed objects $X_1,\dots,X_n$, there will be $K_n \le n$ distinct entities, or species, representing the $n$th value of the accumulation curve. The values $(X_n)_{n \ge 1}$ are randomly generated in a sequential manner, so that the tag $X_{n+1}$ is either new or equal to one of the previously observed objects. For instance, in the Dirichlet process case, the sequential allocation mechanism  for any $n \ge 1$ proceeds as follows:
\begin{equation}\label{eq:Dir_scheme}
(X_{n+1} \mid X_1,\dots,X_n) =
\begin{cases}
\textrm{``new''}, & \textrm{with probability} \quad \alpha/(\alpha+n),\\
X_i, & \textrm{with probability} \quad 1/(\alpha+n),\quad (i =1, \ldots, n),
\end{cases}
\end{equation}
where $\alpha > 0$ controls the rate of new discoveries; see also \citet{Blackwell1973}.  

The predictive scheme in~\eqref{eq:Dir_scheme} is restrictive in depending on a single parameter and in inducing a logarithmic growth for the accumulation curve~$(K_n)_{n \ge 1}$. These limitations motivated the development of more general random processes that allow for polynomial growth rates.  These include the two parameter Poisson--Dirichlet process of \citet{Perman_1992}, often called the Pitman--Yor process when the number of species is assumed to be infinite or 
the Dirichlet-multinomial process in the finite case  \citep{Pitman1997}, and the general classes of Gibbs-type priors \citep{Gnedin2005} and species sampling models \citep{Pitman1996}. The derivation of Bayesian nonparametric estimators for accumulation curves, under general Gibbs-type priors and the Pitman--Yor process, is due to \cite{Lijoi_mena_pruenster_2007} and \citet{Favaro_lijoi_mena_pruenster_2009}, respectively. 

Unfortunately, tractable generalizations of~\eqref{eq:Dir_scheme} such as the Pitman--Yor process are too restrictive for many real-world scenarios. This is evident from Figure~\ref{fig:dir_failure}, which shows in- and out-of-sample performance in estimating the number of distinct fungi species in a given number of fungal \textsc{dna}-barcode sequences. `Species' are defined in this article based on genetic sequences being sufficiently distinct, but the terminology used by ecologists is `operational taxonomic units' as determining species requires additional verification.
The Dirichlet process fails badly in sample, while the Pitman--Yor has good in-sample fit but poor out-of-sample predictive accuracy.  This is not surprising, as the Pitman--Yor process depends on only two parameters and assumes that $K_n \rightarrow \infty$ almost surely as $n\rightarrow \infty$.  As there are finitely many fungi species, $K_n$ should more realistically converge to a finite constant. The Dirichlet-multinomial process allows finite $\lim_{n \rightarrow \infty} K_n = K_{\infty}$ but the trajectory has similar lack of fit as the Dirichlet process.

Potentially one could use a species sampling model that is more flexible than the Pitman--Yor, while also allowing finite $K_\infty$; recent examples include \cite{Camerlenghi2018, Lijoi2020}. However, such specifications involve cumbersome combinatorial structures in the sampling mechanism, effectively preventing their application in the types of large datasets that are now routinely collected in our motivating application areas.  For example, in fungi biodiversity studies, it is now  common to sequence millions of \textsc{dna} barcodes from 10,000s of species \citep[e.g. ][]{Ovaskainen_2020}.

\begin{figure}
\centering\includegraphics[width=0.95\linewidth]{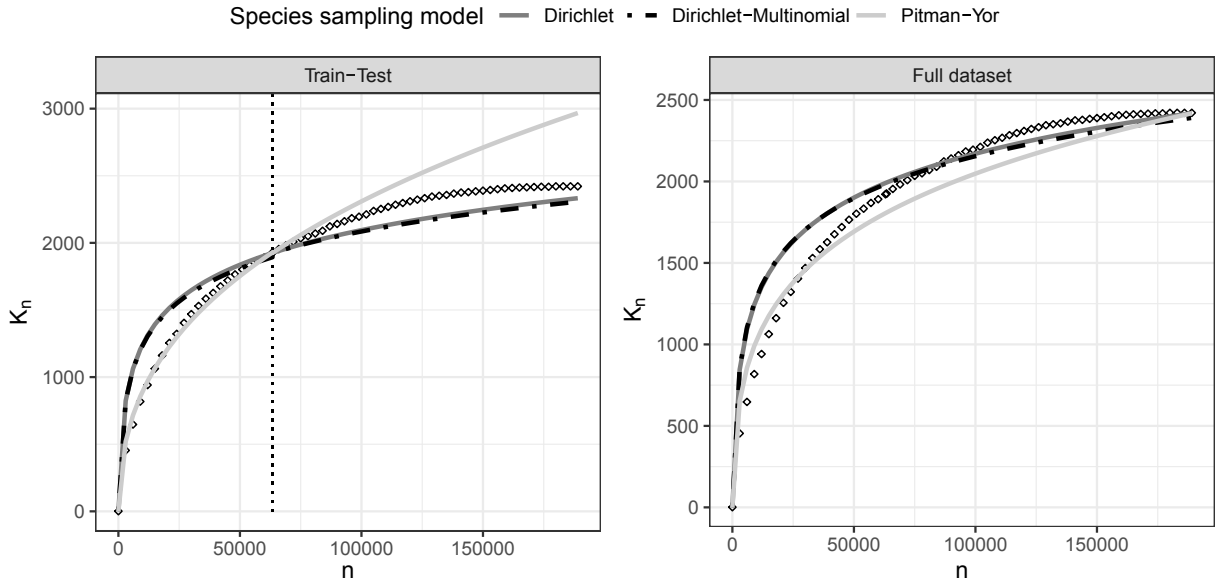}
    \caption{Empirical and estimated accumulation curve in one air fungal \textsc{DNA}-barcoding sample obtained from Finland. White dots indicate the observed values. Left panel: the vertical line is the train-test set cutoff, set to $1/3$ of the total number of genetic sequences. The parameters of the Dirichlet, the Pitman--Yor and the Dirichlet-multinomial processes are based on the training set. Right panel: the curves are instead estimated using the full dataset}
    \label{fig:dir_failure}
\end{figure}

We address the above limitations through a novel modelling framework, 
 which is highly flexible, analytically tractable, and computationally efficient. The key distinction compared to species sampling models, such as~\eqref{eq:Dir_scheme}, is that we directly specify a model for the accumulation curve $(K_n)_{n \ge 1}$, whereas the tags $(X_n)_{n \ge 1}$ are regarded as nuisance parameters. Specifically, we consider a collection of Bernoulli random variables $(D_n)_{n \ge 1}$ representing whether at the 
$(n+1)$th step a new entity has been discovered or not, namely
\begin{equation*}
    \textrm{pr}(D_{n+1} = 1) = \textrm{pr}(X_{n+1} = \text{``new"} \mid X_1,\dots,X_n), \qquad n \ge 1,
\end{equation*}
having set $D_1 = 1$. The accumulation curve is obtained by summing over these binary indicators:
  $  K_n = \sum_{i=1}^nD_i, n \ge 1.$
Differently from general species sampling models, in our framework, the Bernoulli indicators $(D_n)_{n \ge 1}$ are assumed to be \emph{independent}, albeit not identically distributed. Hence, we aim at developing suitable formulations for the probabilities~$(\pi_n)_{n \ge 1}$, with $\pi_n = \mathrm{pr}(D_n = 1)$, for any $n \ge 1$. It is natural to require these probabilities to be decreasing over $n$, so that the discovery of a new entity is increasingly difficult the more data we collect. Moreover, $\pi_1 = \mathrm{pr}(D_1 = 1) = 1$, since the first entity of the sequence is necessarily new. Both requirements are satisfied by the Dirichlet process, where $\pi_n = \alpha / (\alpha + n - 1)$. We propose a general strategy for the specification of $(\pi_n)_{n \ge 1}$, relying on the notion of survival functions, and  study the impact of specific choices on the asymptotic behavior of $K_n$. 

A specific subclass of our framework is particularly appealing in terms of analytic and computational simplicity, due to connections with logistic regression.  This subclass includes the Dirichlet process and naturally leads to covariate-dependent extensions.  Existing covariate-dependent species sampling models are typically complex to implement; refer to  \citet{Quintana2020} for a recent overview.  In contrast, our approach simply involves implementing a logistic regression with certain constraints on the parameters.
We illustrate the flexibility and computational tractability through application to fungi biodiversity data collected at different sampling sites in Finland under different ecological conditions \citep{Abrego_2020}.



\section{A general modelling framework for accumulation curves}\label{Sec:generalFrame}

\subsection{Background on species sampling models}\label{subsec:BNP}

In this Section we review key concepts about  species sampling models that will be used throughout the paper. For a broader overview, refer to \citet{Pitman1996} and \citet{DeBlasi2015}.

Let $(X_n)_{n \ge 1}$ be a sequence of objects. Given the discrete nature of the data, there will be ties among  $X_1,\dots,X_n$, comprising a total of $K_n = k$ distinct entities $X_1^*,\dots,X_k^*$, having frequencies $n_1,\dots,n_k$, with $\sum_{j=1}^kn_j = n$. Species sampling models generalize the sequential allocation of the Dirichlet process in~\eqref{eq:Dir_scheme}, so that  for any $n \ge 1$,
\begin{equation}\label{eq:ssm_scheme}
(X_{n+1} \mid X_1,\dots,X_n) =
\begin{cases}
\textrm{``new''}, & \textrm{with probability} \quad p_{k+1}(n_1,\dots,n_k),\\
X_j^*, & \textrm{with probability} \quad p_j(n_1,\dots,n_k),\quad (j =1, \ldots, k),
\end{cases}
\end{equation}
for suitable probabilities $\sum_{j=1}^{k + 1} p_j(n_1,\dots,n_k) = 1$ and with $X_1 = \text{``new''}$. 
The discovery probabilities $\textrm{pr}(X_{n+1} = \text{``new''} \mid X_1,\dots,X_n)$ depend on the previous values
only through $k$ and/or the 
frequencies $n_1,\dots,n_k$.  Equation~\eqref{eq:ssm_scheme} only leads to a valid species sampling model if the resulting law of $(X_n)_{n \ge 1}$ is \emph{exchangeable}; this is not automatic  as discussed in \citet{Lee2013}. 
In the Pitman--Yor case,  $p_{k+1}(n_1,\dots,n_k) = (\alpha + k\sigma)/(\alpha+ n)$ and $p_j(n_1,\dots,n_k) = (n_j - \sigma)/(\alpha + n)$, for $j=1,\dots,k$ and with $\sigma \in [0,1)$ and $\alpha > -\sigma$; the Dirichlet process is recovered with $\sigma = 0$. The Dirichlet-multinomial has the same sampling scheme of the Pitman--Yor, having set $\sigma<0$ and $\alpha = H |\sigma|$, with $H\in \mathds{N}$ representing the total number of species. 

The sequential mechanism in~\eqref{eq:ssm_scheme} induces a law for the accumulation curve~$(K_n)_{n\ge 1}$. For the remainder of the Section, we focus on the Dirichlet process, since it provides a special case of our framework. In this case, the distribution of $K_n$ is available in closed form \citep{antoniak1974},
\begin{equation}\label{eq:Dir_Kndist}
\mathrm{pr}(K_n =k) = \frac{\alpha^k}{(\alpha)_n} \vert s(n,k)\vert,\qquad \mbox{for any $n \ge 1$},
\end{equation}
where $\vert s(n,k) \vert$ is the signless Stirling number of the first kind; see \citet{Charalambides_2005}.  Moreover, the expectation of \eqref{eq:Dir_Kndist} is
\begin{equation}\label{eq:Dir_E}
E(K_n) = \sum_{i=1}^n\frac{\alpha}{\alpha + i - 1},
\end{equation}
which provides the prior mean for the accumulation curve.  One may also be interested in the posterior distribution of $K_m^{(n)}$, the number of new entities in a future sample of size $m$ conditioning on training data
$X_1,\dots,X_n$. Under a Dirichlet process 
this distribution does not depend on either the past values $X_1,\dots,X_n$ or the observed number of distinct values $K_n = k$; see \citet{Lijoi_mena_pruenster_2007}. Consequently, we obtain
\begin{equation}\label{eq:Dir_prediction}
E(K_m^{(n)} \mid X_1,\dots,X_n) = \sum_{i=1}^m \frac{\alpha}{\alpha + n + i - 1}.
\end{equation}
The Dirichlet process is the only species sampling model for which such a simplification occurs \citep{Lijoi_mena_pruenster_2007}. For example, in the Pitman--Yor process the posterior distribution of $K_m^{(n)}$ depends  on the observed number of distinct values $K_n = k$. 

\subsection{The model}\label{subsec:PoissBin}

In species sampling models, the distribution of the accumulation curve $(K_n)_{n \ge 1}$ is essentially a byproduct of the specification for the values $(X_n)_{n \ge 1}$. Instead, we propose a more direct formulation for $(K_n)_{n \ge 1}$ which avoids modelling of the sequence $(X_n)_{n \ge 1}$. 

Let $(D_n)_{n \ge 1}$ be a collection of \emph{independent} binary indicators, denoting the discoveries, with probabilities $(\pi_n)_{n \ge 1}$. Moreover, let $K_n = \sum_{i=1}^nD_i$ for any $n \ge 1$ be the accumulation curve. We have $D_n = K_n - K_{n-1}$ for any $n \ge 2$ with $D_1 = 1$. Hence, the discoveries $(D_n)_{n \ge 1}$ and the accumulation curve $(K_n)_{n \ge 1}$ carry the same information, having a one-to-one relationship. The resulting distribution of $K_n$ is  Poisson-binomial with parameters $\pi_1,\dots,\pi_n$, which we denote as $K_n \sim \textsc{pb}(\pi_1,\dots,\pi_n)$. As previously discussed, the probabilities $(\pi_n)_{n \ge 1}$ 
must satisfy $\pi_n > \pi_{n+1}$ for every $n \ge 1$ with $\pi_1=1$. 
 In addition, we require $\lim_{n \rightarrow \infty}\pi_n = 0$,  meaning that the probability of making a new discovery should eventually approach zero.
A general strategy for constructing a set of probabilities satisfying these requirements is described as follows. 

\begin{definition}\label{def1} Let $T$ be a random variable on $(0,\infty)$ with strictly increasing cumulative distribution function $F(t; \theta)$ indexed by $\theta \in \Theta \subseteq \mathds{R}^p$. Moreover, let $S(t;\theta) = 1 - F(t; \theta)$ be its survival function. The set of probabilities $(\pi_n)_{n \ge 1}$ are said to be directed by $S(t;\theta)$ if 
\begin{equation}\label{eq:pi_S}
    \pi_n = \textup{pr}(T_n > n - 1) = S(n-1; \theta),\qquad 
    \mbox{for any $n \ge 1$},
\end{equation}
where $(T_n)_{n \ge 1}$ are independent and identically distributed random variables following $F(t ; \theta)$.
\end{definition}

It is easy to check that a set of probabilities $(\pi_n)_{n \ge 1}$ directed by $S(t;\theta)$ satisfies the aforementioned requirements. Indeed, one has that $\pi_1 = S(0;\theta) = 1$ for any $\theta \in \Theta$, since $T$ is supported on $(0,\infty)$. Moreover,  $\pi_n = S(n-1; \theta) > S(n; \theta) = \pi_{n+1}$, because by assumption $S(t;\theta)$ is strictly decreasing. Furthermore, one has that $\lim_{n \rightarrow \infty}\pi_n = \lim_{n \rightarrow \infty} S(n-1; \theta) = 0$, as desired, since $S(t; \theta)$ is a survival function.  Each binary random variable $D_n$ may be represented as $D_n = \mathds{1}(T_n > n - 1)$, with $\mathds{1}(\cdot)$ denoting the indicator function. 

If the probabilities $(\pi_n)_{n \ge 1}$ are directed by $S(t;\theta)$, then inferential statements about the parameter vector $\theta \in \Theta$ can be based on the likelihood function $\mathscr{L}(\theta \mid D_1,\dots,D_n)$ or, equivalently, on $\mathscr{L}(\theta \mid K_1,\dots, K_n)$. The former is readily available as
\begin{equation}\label{eq:LikBern}
    \mathscr{L}(\theta \mid D_1,\dots,D_n) \propto \prod_{i=2}^{n}S(i-1;\theta)^{D_i} \{1-S(i-1;\theta)\}^{1-D_i},
\end{equation}
having excluded the degenerate term $D_1 = 1$. The Dirichlet process implicitly assumes  $S(t; \alpha) = \alpha/(\alpha+t)$ with $\alpha > 0$, which is the survival function of a continuous random variable $T$, as will be shown in Section~\ref{Sec:logistic}. Since the Dirichlet process is both a species sampling model and a member of our general framework, it is interesting to compare the information contained in the data $X_1,\dots,X_n$ with that carried by the discovery indicators $D_1,\dots,D_n$. This can be formally achieved by comparing the likelihood functions $\mathscr{L}(\alpha \mid X_1,\dots, X_n)$ and $\mathscr{L}(\alpha \mid D_1,\dots, D_n)$. The former may be obtained following \citet{antoniak1974}, whereas the latter coincides with \eqref{eq:LikBern}, having set $S(t; \alpha) = \alpha/(\alpha+t)$. Before stating the result, let us denote with $(a)_n = a(a+1)\cdots(a + n -1)$ the Pochhammer symbol, for any $a > 0$ and $n \ge 1$.

\begin{theorem}\label{theo:EPPF_Bernoulli} Let $(X_n)_{n \ge 1}$ be a sequence of objects directed by a Dirichlet process as in~\eqref{eq:Dir_scheme} and let $(D_n)_{n \ge 1}$ be the associated discovery indicators. Then for a sample $X_1,\dots,X_n$ with $K_n = k$ distinct values one has
\begin{equation*}
    \mathscr{L}(\alpha \mid D_1,\dots,D_n) \propto \mathscr{L}(\alpha \mid X_1,\dots, X_n) \propto \frac{\alpha^k}{(\alpha)_n}.
\end{equation*}
\end{theorem}

Hence, it is equivalent to base inferences on the Dirichlet process parameter $\alpha$ on the likelihood ~\eqref{eq:LikBern} for the discovery indicators instead of the usual likelihood for $X_1,\dots,X_n$. Broadly speaking, this occurs because $K_n = \sum_{i=1}^nD_i$ is the minimal sufficient statistic for $\alpha$ in the Dirichlet process; see also \citet{Lijoi_mena_pruenster_2007} for similar considerations. An implication is that the empirical Bayes estimate of $\alpha$, obtained by maximizing $\alpha^k/(\alpha)_n$, coincides with the maximizer of  \eqref{eq:LikBern}.

\begin{remark}
The Dirichlet process in~\eqref{eq:Dir_scheme} is the only species sampling model in which the discovery indicators $(D_n)_{n \ge 1}$ are independent \citep{zabell1982, Lee2013}. Thus, our framework can be regarded as the result of some generative mechanism such as~\eqref{eq:ssm_scheme}, but the underlying sequence $(X_n)_{n \ge 1}$ will not necessarily be exchangeable. The main implication is that inference on the parameter $\theta \in \Theta$ may depend on the order of the observations, which might be reasonable or not depending on the application. However, the lack of exchangeability has mild practical implications, as illustrated in the following.
\end{remark}

\subsection{Smoothing, prediction and posterior representations}

In this Section, we present prior and posterior properties of $K_n$, which may be useful for both smoothing and prediction. Supposing $(\pi_n)_{n \ge 1}$ is directed by $S(t;\theta)$, then  $K_n \sim \textsc{pb}\{1, S(1;\theta), \ldots, S(n-1;\theta)\}$, with prior mean and variance equal to 
\begin{equation*}
E(K_n) = \sum_{i=1}^n S(i-1; \theta), \quad  \mathrm{var}(K_n) = \sum_{i=1}^n S(i-1;\theta)\{1-S(i-1;\theta)\},  \qquad n \ge 1. 
\end{equation*}
These moment formulas may be useful in choosing the parametric form of $S(t;\theta)$ and for prior elicitation for $\theta$. 
In ecology, in-sample estimation of species accumulation curves is sometimes called {\em rarefaction}; this amounts to smoothing of the $K_1,\ldots,K_n$ values observed in the training samples.  In our 
framework, 
$E(K_n)$ is the sum of the first $n$ discovery probabilities
\begin{equation*}
   E(K_n) = \sum_{i=1}^n \pi_i = \sum_{i=1}^n \textrm{pr}(D_i = 1), \qquad n \ge 1.
\end{equation*}
This expectation does not depend on the ordering of the data, at least for any fixed value of $\theta$.

Suppose we are given a sample of $D_1,\dots,D_n$ discoveries displaying $K_n = k$ distinct entities and that we are interested in predicting future values of the accumulation curve $K_{n+1}, \ldots, K_{n+m}$ or in predicting the number of new entities within a future sample of size $m$,  $K_m^{(n)} = K_{n + m} - K_n = \sum_{i = n + 1}^{n + m}D_i$. The posterior distribution of $(K_m^{(n)} \mid D_1,\dots,D_n)$ is available in closed form, as summarized in the next Proposition, immediately leading to the distribution of $(K_{n + m} \mid D_1,\dots,D_n) = k + K_m^{(n)}$.

\begin{proposition}\label{pro:PoisPrediction} Let $(D_n)_{n \ge 1}$ be a collection of independent discovery indicators with probabilities $(\pi_n)_{n \ge 1}$ directed by $S(t;\theta)$. Moreover, let $K_m^{(n)} = K_{n+m} - K_n$. Then for any $n \ge 1$ and $m > n$
\begin{equation*}
 (K_m^{(n)} \mid D_1,\dots,D_n) \sim \textsc{pb}\{S(n; \theta), \ldots, S(n + m -1;\theta)\}.   
\end{equation*}
Hence, it follows that
\begin{equation*}\label{eq:PB_predicted_val}
   E(K_m^{(n)} \mid D_1,\dots,D_n) = \sum_{i = n+1}^{n + m}\textup{pr}(D_i = 1) = \sum_{j=1}^{m} S(j + n-1; \theta),
\end{equation*}
implying that $E(K_{n +m} \mid D_1,\dots,D_n) = k + E(K_m^{(n)} \mid D_1,\dots,D_n)$.
\end{proposition}

Within the ecological community,  out-of-sample prediction through $E(K_{n +m} \mid D_1,\dots,D_n)$ is sometimes called \emph{extrapolation}, which again can be interpreted as a sum of discovery probabilities. 
Proposition~\ref{pro:PoisPrediction}
implies that the  posterior distribution of $K_m^{(n)}$ is conjugate, being a Poisson-binomial distribution with updated parameters. In addition, the posterior law of $K_{n + m}$ only depends on $K_n = k$, meaning that extrapolation of the accumulation curve also does not depend on the order of $D_1,\dots,D_n$ for any fixed value of $\theta$.

\subsection{Asymptotic behavior of $K_n$}\label{subsec:asymptotic}

The limit of $K_n$ as $n \rightarrow \infty$ is often of inferential interest, representing the random number of entities one would eventually discover. Depending on the choice of $S(t;\theta)$, two scenarios can occur: i) the number of distinct entities diverges, as in the Dirichlet process case, so that $K_n \rightarrow \infty$ almost surely as $n \rightarrow \infty$. In this regime, it is useful to study the growth rate of $K_n$. Alternatively, we could find that ii) the number of distinct species converges to some non-degenerate random variable $K_n \rightarrow K_\infty$, almost surely, as $n \rightarrow \infty$. Within ecology the random variable $K_\infty$ is called the \emph{species richness}.

The asymptotic behaviour of $K_n$ is controlled by the structure of the chosen survival function~$S(t;\theta)$. Before stating our first result, let  us define
 $E(T) = \int_{0}^{\infty} \mathrm{pr}(T> t)\mathrm{d}t =  \int_{0}^{\infty} S(t; \theta)\mathrm{d}t,$
that is, the expectation of the latent variables in Definition~\ref{def1}. 

\begin{proposition}\label{pro:ET} Let $K_n \sim \textsc{pb}\{1,S(1;\theta),\dots,S(n -1;\theta)\}$. Then, there exists a possibly infinite random variable $K_\infty$ such that 
$\lim_{n \rightarrow \infty}    K_n \to K_\infty,$
almost surely, with $E(K_\infty)=\sum_{i=0}^{\infty}S(i)$. Moreover, 
\begin{equation}\label{eq:ET_ineq}
    E(T) \le E(K_\infty) \le E(T) + 1.
\end{equation}
\end{proposition}
Equation~\eqref{eq:ET_ineq} provides lower and upper bounds for the asymptotic mean, which can be used to summarize the species richness. Besides, the expected value of $E(T)$  represents a simple tool to determine whether the accumulation curve diverges or not, as the following Corollary clarifies.

\begin{corollary}\label{cor:ET_inequalities}
Under the conditions of Proposition~\ref{pro:ET}, $K_\infty=\infty$ almost surely if and only if $E(T) = \infty$. 
\end{corollary}

Let us consider the first asymptotic regime, corresponding to the  $K_\infty = \infty$ case. In this case, the rate of growth is controlled by $S(t;\theta)$, as clarified in the following Theorem, which also presents a central limit approximation.

\begin{theorem}\label{theo:CLT}
Let $K_n \sim \textsc{pb}\{1,S(1;\theta),\dots,S(n -1;\theta)\}$ and suppose $K_\infty = \infty$ almost surely.  Then, as $n\rightarrow \infty$, 
$K_n/b_n \to 1$ almost surely, for 
$b_n = \int_1^n S(t - 1; \theta)\mathrm{d}t$.
In addition, 
\begin{equation*}
 \frac{K_n-E(K_n)}{\mathrm{var}(K_n)^{1/2}}\to N(0,1), \qquad n \rightarrow \infty,
\end{equation*}
in distribution.
\end{theorem}

Theorem \ref{theo:CLT} implies that the growth rate of $K_n$ corresponds to $b_n = \int_1^n S(t - 1; \theta)\mathrm{d}t$. In the Dirichlet process case, $b_n =  \alpha \log{(\alpha + n - 1)} - \alpha \log{\alpha}$, corresponding to the well-known growth rate $\alpha \log{n}$ \citep{Ramesh_Korwar_1973}.  The $N(0,1)$ limiting distribution allows one to assess uncertainty in $K_n$ for large $n$.


Consider now the second asymptotic regime, namely the  $K_\infty < \infty$ case. 
Although the distribution of $K_{\infty}$ is generally not available in closed form, the first two moments are well defined. 

\begin{corollary}\label{cor:ET_var}
Under the conditions of Proposition~\ref{pro:ET}, if $K_\infty < \infty$ almost surely, then $E(K_\infty) = \sum_{i=1}^\infty S(i-1;\theta) < \infty$ and $\textup{var}(K_\infty) = \sum_{i=1}^\infty S(i-1;\theta)\{1 - S(i-1;\theta)\} < \infty$. \end{corollary}
Hence, a natural estimator for the  species richness is $E(K_\infty)$, which may be numerically approximated; for instance by truncating the infinite summation  $E(K_\infty) = \sum_{i=0}^\infty S(i)$. Alternatively, one could exploit equation~\eqref{eq:ET_ineq}
and consider the arithmetic mean of the bounds, obtaining the approximation
$E(K_\infty) \approx E(T) + 1/2,$
which is often easier to compute than $E(K_\infty)$ and is highly accurate when the number of species is not small.
Despite the absence of a central limit theorem in this case, there exist several approximations for Poisson-binomial distributions, which may be used for $K_\infty$ and $K_n$ when $n$ is large  \citep{Hong_2013}. 

Interestingly, Proposition~\ref{pro:PoisPrediction} offers a natural estimator for the posterior species richness as well, namely $E(K_\infty \mid D_1, \ldots, D_n)$. Consider $E(K_{m+n} \mid D_1, \ldots, D_n)$ and let $m \to \infty$. Then, it is straightforward to see that  $E(K_\infty \mid D_1, \ldots, D_n) = k + E(K_\infty^{(n)} \mid D_1, \ldots, D_n)$, where 
$$E(K_\infty^{(n)} \mid D_1, \ldots, D_n) = \sum_{j=1}^\infty S(j+n-1;\theta).$$
Hence, all the properties of $K_{\infty}$ 
can be naturally extended to the posterior species richness.

\section{Logistic models}\label{Sec:logistic}

\subsection{The log-logistic distribution}\label{subsec:loglogistic}

The framework in the previous Section requires elicitation of  $S(t; \theta)$. In this Section, we focus on a class of survival functions, which lead to a generalization of the Dirichlet process, enjoy appealing analytical and computational properties and result in natural covariate-dependent extensions, as described in Section~\ref{subsec:Logistic_regr}.  In particular, we first consider a two parameter case 
\begin{equation}\label{eq:LogLogistic}
    S(t; \alpha, \sigma) = \frac{\alpha}{\alpha + t^{1-\sigma}}, \qquad t \ge 0,
\end{equation}
where $\alpha > 0$ and $\sigma < 1$. The survival function $S(t; \alpha, \sigma)$ characterizes a two-parameter log-logistic distribution, and therefore we will write $T \sim \textsc{ll}(\alpha,\sigma)$. 
Clearly, when $\sigma = 0$, $S(t; \alpha, 0)$ reduces to the Dirichlet process case. The parameter $\sigma$ plays a similar role to the discount parameter of the Pitman--Yor process and general Gibbs-type priors. 
For any $\sigma < 0$, one has
\begin{equation*}
E(T) =  
\frac{\alpha^{1/(1-\sigma)}\pi}{(1-\sigma)\sin\{\pi/(1-\sigma)\}},
\end{equation*}
implying that when $\sigma < 0$ the limiting distribution $K_\infty < \infty$ is non-degenerate, thanks to Corollary~\ref{cor:ET_inequalities}. Conversely, when $0 \leq \sigma <1$, one has that both $K_\infty = \infty$ and $E(T) = \infty$. The rate at which this occurs is logarithmic 
in the Dirichlet process case in which $\sigma = 0$. In contrast, for $\sigma > 0$, one can show that the growth of $K_n$ is polynomial, so that in the notation of Theorem
~\ref{theo:CLT} one has $b_n = \int_1^n S(t; \alpha, \sigma) \mathrm{d}t = \mathcal{O}(n^\sigma)$. These considerations reinforce the parallelism with Gibbs-type priors; see \citet{Gnedin2005} and \citet{ DeBlasi2015} for details.

In the next Section, we describe a three-parameter extension of the log-logistic distribution and derive combinatorial tools and distributional properties that
also apply to $S(t; \alpha, \sigma)$ in~\eqref{eq:LogLogistic}. 


\subsection{A three parameter log-logistic distribution}\label{subsec:augloglogistic}

In this Section we extend the log-logistic specification by including an additional parameter, denoted as $\phi$, which forces $K_n$ to converge to a non-degenerate distribution. This allows us to restrict focus to the second asymptotic regime. In particular, we let $\theta = (\alpha, \sigma, \phi)$ and \begin{equation}\label{eq:AugLoglogistic}
    S(t; \alpha, \sigma, \phi) = \frac{\alpha \phi^{t}}{\alpha \phi^{t} + t^{1-\sigma}}, \qquad t \ge 0,
\end{equation}
with $\alpha > 0$, $\sigma<1$  and $0 < \phi \le 1$. The two parameter specification is recovered when  $\phi = 1$. We call the distribution of $S(t;\alpha,\sigma, \phi)$ a three-parameter log-logistic, written $T \sim \textsc{ll}(\alpha, \sigma, \phi)$. 
\begin{proposition}\label{prop:convergence}
Let $K_n \sim \textsc{pb}\{1,S(1;\theta),\dots,S(n -1;\theta)\}$, with $S(t;\theta)$ defined as in Equation~\eqref{eq:AugLoglogistic}. Then for any $0<\phi < 1$ it holds that $K_n \to K_\infty < \infty$ almost surely as $n \to \infty$.
\end{proposition}

 Proposition~\ref{prop:convergence} ensures that for $0<\phi < 1$ the species richness is always finite. 
For the remainder of the Section, we discuss some combinatorial properties related to the law of $K_n$. While having their own theoretical relevance, our results facilitate computation of the probability mass function of $K_n$ and draw further parallels with Gibbs-type priors.


\begin{definition}\label{def2} Let $\alpha > 0$, $\sigma < 1$ and $0 < \phi \le 1$. Then for any $n \ge 1$ and $0 \le k \le n$ we define $\mathscr{C}_{n,k}(\sigma, \phi)$ as the coefficients of the polynomial expansion
$\prod_{k=0}^{n-1}(\alpha + k^{1-\sigma}\phi^{-k})=\sum_{k=0}^n\alpha^k\:\mathscr{C}_{n,k}(\sigma, \phi),$
having set $\mathscr{C}_{0,0}(\sigma, \phi) = 1$.
\end{definition}

In the special case $\phi =1$ and $\sigma = 0$ one recovers the definition of the signless Stirling numbers of the first kind, namely $\mathscr{C}_{n,k}(0,1) = \vert s(n,k) \vert$; see \citet{Charalambides_2005}. In addition, the coefficients $\mathscr{C}_{n, k}(\sigma, \phi)$ can be conveniently computed through recursive formulas.

\begin{theorem}\label{teo:C_recursion}
The coefficients $\mathscr{C}_{n, k}(\sigma, \phi)$ of Definition~\ref{def2}  satisfy the triangular recurrence  
\begin{equation*}\label{eq:C_nk_recursion}
    \mathscr{C}_{n+1,k}(\sigma, \phi) = \mathscr{C}_{n,k - 1}(\sigma, \phi) + n^{1-\sigma}\phi^{-n}\mathscr{C}_{n,k}(\sigma, \phi),
\end{equation*}
for any $n \ge 0$ and $1 \le k \le n+1$, with initial conditions
$\mathscr{C}_{0,0}(\sigma, \phi) = 1,$  $\mathscr{C}_{n,0}(\sigma, \phi) = 0,$
$n \ge 1,$ 
$\mathscr{C}_{n,k}(\sigma, \phi) = 0,$  
$k > n.$
Moreover, for any $1 \le k \le n$ and $n \ge 2$, one has \begin{equation*}\label{eq:C_nk}
    \mathscr{C}_{n, k}(\sigma, \phi) = \sum_{(i_1, \ldots, i_{n-k})}\prod_{j=1}^{n-k}i_j^{1-\sigma}\phi^{-i_j},
\end{equation*}
where the sum runs over the $(n - k)$-combinations of integers $(i_1,\dots,i_{n-k})$ in $\{1,\dots,n-1\}$.
\end{theorem}

We can now state the main theoretical result, namely the probability mass function of $K_n$, which can be expressed in terms of the coefficients $\mathscr{C}_{n,k}(\sigma,\phi)$. 

\begin{theorem}\label{theo:Kn_pmf}
Let $K_n \sim \textsc{pb}\{1, S(1; \alpha, \sigma, \phi), \ldots, S(n-1; \alpha, \sigma, \phi)\}$ for every $n \geq 1$. Then, 
\begin{equation*}\label{eq:Kn_dist}
    \mathrm{pr}(K_n = k) = \frac{\alpha^k}{\prod_{i=0}^{n-1}(\alpha + i^{1-\sigma}\phi^{-i})}\mathscr{C}_{n,k}(\sigma, \phi).
\end{equation*}
\end{theorem}
Theorem~\ref{theo:Kn_pmf} reduces to the distribution obtained by \citet{antoniak1974} and recalled in equation~\eqref{eq:Dir_Kndist} when $\sigma = 0$ and $\phi = 1$. Gibbs-type priors enjoy a similar structure for the distribution of $K_n$, having replaced $\mathscr{C}_{n,k}(\sigma,\phi)$ with the so-called generalized factorial coefficients, which have similar properties; see \citet{Gnedin2005, DeBlasi2015} for further discussion. 

\subsection{Covariate-dependent models}\label{subsec:Logistic_regr} 

Under the three parameter log-logistic specification, the discovery probabilities are
$\pi_{n+1} = \mathrm{pr}(D_{n+1} = 1) = \alpha \phi^{n}(\alpha \phi^{n} + n^{1-\sigma})^{-1}$ for $n\geq 1$ with $\pi_1 = 1$. An interesting and practically useful property of our model is the following representation
\begin{equation}\label{eq:Logistic_regres}
\log \frac{\pi_{n+1}}{1-\pi_{n+1}} = \log \alpha - (1-\sigma)\log n + (\log \phi)n = \beta_0 + \beta_1\log{n} + \beta_2 n, \qquad n \ge 1,
\end{equation}
having set $\beta_0 = \log \alpha$, $\beta_1 = \sigma - 1 < 0$ and $\beta_3 = \log \phi \le  0$. Hence, equation~\eqref{eq:Logistic_regres} has the form of  a logistic regression for the binary indicators $D_2, \ldots, D_n$, where the regression coefficients $\beta_2$ and $\beta_3$ are constrained to be negative.  By letting $\beta_1 = -1$ and $\beta_2 = 0$ one recovers the discovery probability of the Dirichlet process. This representation has computational advantages, which will be discussed in Section~\ref{Sec:Inference}.

The logistic regression representation in~\eqref{eq:Logistic_regres} suggests natural extensions to accommodate covariates.
 Suppose we are given a collection of $L$ accumulation curves, namely $(K_{1 n})_{n \ge 1}, \dots,(K_{L n})_{n \ge 1}$, representing for example the sequential discoveries recorded at different geographical locations. Each location is associated with a set of covariates $z_\ell^\mathrm{T}= (z_{\ell 1}, \ldots, z_{\ell p}) \in \mathds{R}^p$ for $\ell =1,\dots,L$. Let $(D_{\ell n})_{n \ge 1}$ be the sequence of discovery indicators for the $\ell$th location, with probabilities $(\pi_{\ell n})_{n \ge 1}$. The most flexible specification for $K_{\ell n}$ corresponds to the case in which all the parameters are location-specific, so that for any $n \ge 1$,
\begin{equation*}
\log \frac{\pi_{\ell n+1}}{1-\pi_{\ell n+1}} =  \beta_{\ell 0} + \beta_{\ell 1}\log{n} + \beta_{\ell 2} n, \qquad (\ell =1,\dots,L).
\end{equation*}
This specification can borrow information across locations via a 
hierarchical model on $\beta_l = (\beta_{l0},\beta_{l1},\beta_{l2})^\mathrm{T}$ or by fixing certain parameters.  Alternatively, systematic variation across locations can be modeled through including covariates $z_\ell$ via 
\begin{equation}\label{eq:regress_covariate}
\log \frac{\pi_{\ell n+1}}{1-\pi_{\ell n+1}} =  \beta_{\ell 0} + \beta_{\ell 1}\log{n} + \beta_{\ell 2} n = z_\ell^\mathrm{T} \gamma_0 + (z_\ell^\mathrm{T} \gamma_1) \log{n}  + (z_\ell^\mathrm{T} \gamma_2) n,
\end{equation}
for $\ell = 1,\dots,L$, with $\gamma_0, \gamma_1, \gamma_2 \in \mathds{R}^p$ being vectors of coefficients such that $z_\ell^\mathrm{T} \gamma_2 < 0$ and $z_\ell^\mathrm{T} \gamma_2 \le 0$. This specification is still in the form of a logistic regression and therefore inference on the parameters $\gamma_0,\gamma_1$ and $\gamma_2$ can be conducted through straightforward modifications of standard algorithms. The computational details are discussed in the next Section.

\section{Posterior computation}\label{Sec:Inference}
\subsection{Estimation procedures}



Consider the model in
equation~\eqref{eq:Logistic_regres}. 
The parameters $\theta = (\alpha, \sigma, \phi)$ can be estimated 
by maximizing the likelihood in equation~\eqref{eq:LikBern}, with $S(t; \theta) = S(t; \alpha, \sigma, \phi)$, $\beta_1 < 0$ and $\beta_2 \le 0$. In practice, it may suffice to ignore these constraints and apply routine algorithms for fitting logistic regression, as the unconstrained maximum likelihood estimates typically satisfy the constraints.  In this case, the resulting estimate $\hat{\theta}$ has the following appealing property.
\begin{proposition}\label{pro:regression}
Let $\hat{\theta} = (\hat{\alpha}, \hat{\sigma}, \hat{\phi})$ be the unconstrained maximizer of equation~\eqref{eq:LikBern} under the three-parameter specification in~\eqref{eq:AugLoglogistic}, if it exists. If $K_n = k$ is the number of discoveries within the data $D_1,\dots,D_n$, then the expectation $E(K_n)$, evaluated at $\hat{\theta}$, equals $k$.
\end{proposition}

In other words, the $n$th term of the smoothed accumulation curve $E(K_n)$ matches the total number of distinct labels observed in the sequence when the parameters are estimated through unconstrained maximum likelihood. 


Although we can obtain confidence intervals and standard errors for the parameters using the aforementioned maximum likelihood strategy, conducting inferences in this manner ignores the parameter constraints.  In contrast, a fully Bayesian approach can easily incorporate them through a prior, such as 
$\beta \sim N(\mu, \Sigma)\mathds{1}(\beta_1 < 0; \beta_2 \leq 0).$
Under this prior, a straightforward modification  
of the P\'olya-gamma data-augmentation strategy of \citet{Polson_2012} can be used for posterior sampling. 
The covariate-dependent regression detailed in equation~\eqref{eq:regress_covariate} can be naturally carried out in a similar manner. For additional details refer to the Supplementary Material. 

\subsection{Simulations}\label{subsec:Simulation}

We test our log-logistic models on four synthetic sequences of length $n=90,000$. These sequences are simulated according to four different models: i)  Dirichlet process with $\alpha = 30$, ii)  Pitman--Yor process with $\alpha = 30$ and $\sigma= 0.25$, iii) Dirichlet-multinomial process with $\sigma = -0.25$ and $\alpha = H|\sigma|$ and $H= 5,000$, and iv) drawing the species from a Zipf distribution with support $\{1,\dots, H\}$ with $H= 5,000$ and shape parameter $0.3$. Our log-logistic formulations display excellent performance in each case even though the generating mechanisms are species sampling models. Cases i)-ii) have $K_{\infty}=\infty$ while for iii)-iv) we get $K_{\infty}=5,000$.
 We estimate the parameters of each model on the first third of each sequence, comprising $30,000$ data points, and then assess predictive performance on the remaining $60,000$ observations.


We chose truncated and independent normal priors centered at $0$ and with standard deviation $10$. We run a Markov Chain Monte Carlo algorithm for a total of $15,000$ iterations, discarding the first $5,000$ samples. We estimate out-of-sample accumulation curves by  averaging over the posterior samples of $E(K_{n+m} \mid D_1,\dots,D_n)$, obtained as in Proposition~\ref{pro:PoisPrediction}. To compare different log-logistic specifications, we use the Deviance Information Criterion (\textsc{dic}) of \citet{Spiegelhalter_2002}. In Table~\ref{tab:simulations} the \textsc{dic} values and the absolute deviations between the predicted and the true values of $K_{n+m}$ are reported,  with $n = 30,000$ and $m \in \{n/3, n, 2n\}$. 

\begin{table}[h]
\centering
\def~{\hphantom{0}}
\caption{Model performances over simulated sequences of length 90,000. Estimates are based on the first 30,000 observations.}{%
\begin{tabular}{llccccccc}
&  & \multicolumn{3}{c}{Posterior means} &  & \multicolumn{3}{c}{Average prediction error} \\
Data & Model & $\alpha$ & $\sigma$ & $\phi$ & \textsc{dic} & $m = n/3$ & $m = n$ & $n=2n$  \\[5pt]
Dirichlet & \textsc{ll}-1 & $33.47$ & - & - & $1,983.60$ & $2.38$ & $0.82$ & $1.75$ \\ 
& \textsc{ll}-2& $31.07$ & $0.02$ & - & $1,985.34$ & $1.84$ & $0.57$ & $4.11$ \\ 
\vspace{0.2cm}
 & \textsc{ll}-3& $25.22$ & $0.06$ & $0.99$ & $1,986.59$ & $4.15$ & $6.58$ & $9.93$ \\ 
Pitman-Yor & \textsc{ll}-1& $111.49$ & - & - & $5,187.99$ & $24.03$ & $51.93$ & $96.79$ \\ 
 & \textsc{ll}-2& $24.08$ & $0.20$ & - & $5,130.76$ & $3.37$ & $3.2$ & $0.49$ \\ 
  \vspace{0.2cm}
 & \textsc{ll}-3& $20.67$ & $0.22$ & $0.99$ & $5,132.38$ & $9.13$ & $16.74$ & $44.49$ \\ 
Dir-multinomial & \textsc{ll}-1& $760.92$ & - & - & $14,392.07$ & $63.18$ & $193.96$ & $327.3$ \\ 
 & \textsc{ll}-2& $2,239.38$ & $-0.12$ & - & $14,351.46$ & $19.51$ & $78.05$ & $127.24$ \\ 
  \vspace{0.2cm}
 & \textsc{ll}-3& $1,512.80$ & $-0.07$ & $0.99$ & $14,350.51$ & $4.26$ & $0.51$ & $34.81$ \\ 
Zipf & \textsc{ll}-1& $1,699.21$ & - & - & $18,669.83$ & $446.91$ & $1,104.64$ & $1,777.94$ \\ 
 & \textsc{ll}-2& $9.2\times10^5$ & $-0.71$ & - & $17,495.11$ & $135.63$ & $299.57$ & $429.15$ \\ 
 & \textsc{ll}-3& $5,087.35$ & $-0.01$ & $0.99$ & $17,274.05$ & $17.63$ & $22.80$ & $24.10$ \\ 
\end{tabular}}
\label{tab:simulations}
\end{table}
\vspace{-0.05cm}

 Throughout the rest of the paper, we refer to the Dirichlet process and two- and three-parameter log-logistic models as \textsc{ll}-1, \textsc{ll}-2 and \textsc{ll}-3, respectively. When data are generated according to a Dirichlet process, \textsc{ll}-1 has the lowest  \textsc{dic} and prediction errors. Prediction for \textsc{ll}-2 closely resembles that for \textsc{ll}-1, as the posterior mean for $\sigma$ is close to $0$. 
 When data are generated according to a Pitman--Yor, \textsc{ll}-2 achieves the highest accuracy, as expected,
 since our model also accounts for polynomial growth rates for $K_n$. 
 
When the true number of species is finite, 
our \textsc{ll}-3 model has much better performance compared to  \textsc{ll}-1 and  \textsc{ll}-2. For the sequence from the Dirichlet-multinomial process, there is little difference between \textsc{ll}-2 and \textsc{ll}-3 in terms of \textsc{dic}.
 However, the increased flexibility of  \textsc{ll}-3 leads to a much higher out-of-sample accuracy. This behavior is even more evident in the sequence generated from a Zipf distribution. In this case, all $5,000$ species are observed in the first $30,000$ samples so that the true accumulation curve $K_{n + m}$ is a horizontal line. 

\section{Fungal biodiversity application}\label{Sec:Application}

We analyze data from a fungi biodiversity study in Finland \citep{Abrego_2020}. Each sample contains a large number of fungal \textsc{dna} barcode sequences obtained either from air samples or soil samples. 
The total number of sequences to be obtained by each sample can be controlled for when performing the high-throughput sequencing, with the cost of sequencing increasing with the number of sequences. To optimize the number of sequences per sample, one would like to know how many species are to be missed for a particular sequencing depth. This is the goal of our analysis.


The data consist of 174 different samples from different sites across five cities in Finland. For each site, fungi samples are collected on the same dates at two urban areas, one at the core and one at the edge of the city, and two nearby natural areas,
again with one at the core and one at the edge.
Two different sampling methods were used: i) through air, via a cyclone trap and continuously for 24 hours, and ii) through soil, gathering a small portion of soil close to the air trap. We exclude samples with less than $1,000$ sequences, as in such cases the samples lacked sufficient numbers of spores for more comprehensive barcoding. 

This leaves us with a total of 166 samples; the average number of barcoded \textsc{DNA} sequences per sample is $111,980$ and the average number of species discovered is $1,184$.  We fit the three different log-logistic models to training data containing the first third of the fungal barcode sequences, allowing each sample to have its own parameters.  Model fitting and prediction proceeded exactly as in 
Section~\ref{subsec:Simulation}.  

In Table~\ref{tab:training_test} we report the percentage absolute errors between the predicted and true values $K_n$ and $K_{n+m}$, averaged over the 166 samples. As in Figure~\ref{fig:dir_failure}, model \textsc{ll}-1
obtained poor in-sample fit.
Model \textsc{ll}-2 has better in-sample accuracy but exhibits explosive out-of-sample behavior. Model \textsc{ll}-3 is highly accurate both in- and out-of-sample, while having the lowest \textsc{dic}. This is confirmed by Figure~\ref{fig:LL3_models}, which displays the performance of the \textsc{ll}-3 model on the same data as used in Figure~\ref{fig:dir_failure}. 
In Table~\ref{tab:finalnd_model} we report in-sample performance of the three models using all the available \textsc{DNA} sequences. 
Again, \textsc{ll}-3 has the lowest \textsc{dic}. The \textsc{ll}-1 model strongly over-predicts the initial part of the curve as before. Although \textsc{ll}-2 represents a slight improvement over \textsc{ll}-1, the three-parameter log-logistic \textsc{ll}-3 has uniformly better performance. 

\begin{table}
\centering
\def~{\hphantom{0}}
\caption{Summary of the model performances across the 166 samples. Estimates use the first $1/3$ of the \textsc{DNA} sequences in each sample.}{%
\begin{tabular}{lcccc|ccccc}
&& \multicolumn{7}{c}{Fraction of the curve} \\
 Model & \textsc{dic} & $0.1$ & $0.25$ & $0.33$ & $0.50$ & $0.66$ & $0.75$ & $1.00$ \\
 && \multicolumn{7}{c}{Average percentage errors} \\
 \textsc{ll}-1 & $13.61 \times 10^5$ & $17.06$ & $3.54$ & $0.10$ & $3.31$ & $3.86$ & $3.76$ & $2.83$ \\ 
  \textsc{ll}-2 & $13.19 \times 10^5$ & $3.37$ & $1.91$ & $0.10$ & $4.03$ & $8.58$ & $11.25$ & $20.25$ \\ 
  \textsc{ll}-3 & $13.17 \times 10^5$ & $1.55$ & $0.98$ & $0.10$ & $1.94$ & $2.97$ & $3.5$ & $5.27$ \\
\end{tabular}}
\label{tab:training_test}
\end{table}

\begin{table}
\centering
\def~{\hphantom{0}}
\caption{Summary of the model performances across the 166 samples. Estimates are based on all the available \textsc{DNA} sequences.}{%
\begin{tabular}{lccccccccc}
&& \multicolumn{7}{c}{Fraction of the curve} \\
 Model & \textsc{dic} & $0.1$ & $0.25$ & $0.33$ & $0.50$ & $0.66$ & $0.75$ & $1.00$ \\
 && \multicolumn{7}{c}{Average percentage errors} \\
 \textsc{ll}-1 & $18.29 \times 10^5$ & $18.07$ & $4.94$ & $2.68$ & $3.01$ & $3.31$ & $2.91$ & $0.12$ \\ 
\textsc{ll}-2 & $18.03 \times 10^5$ & $3.34$ & $7.61$ & $8.6$ & $8.35$ & $6.5$ & $5.20$ & $0.16$ \\ 
\textsc{ll}-3 & $17.78 \times 10^5$ & $2.47$ & $2.13$ & $1.38$ & $0.67$ & $1.00$ & $1.06$ & $0.12$ \\ 
\end{tabular}}
\label{tab:finalnd_model}
\end{table}

\begin{figure}
\centering
\includegraphics[width=\linewidth]{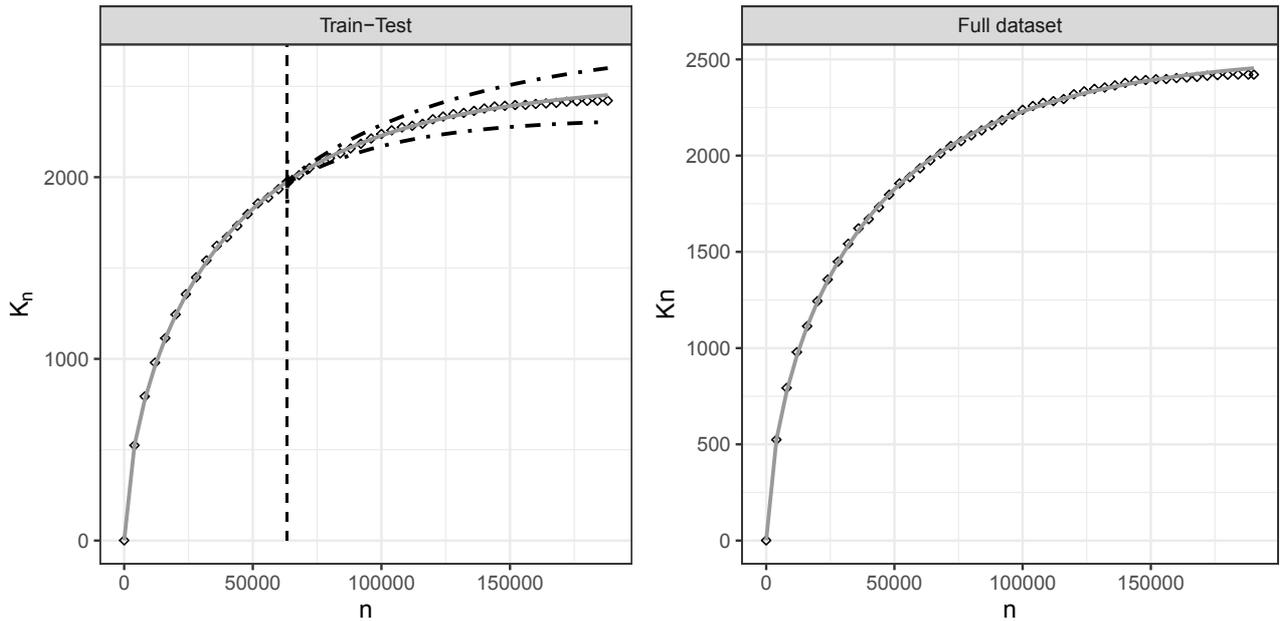}
    \caption{Performance of \textsc{ll}-3 on the same data as Figure~\ref{fig:dir_failure}. The dots represent the observed values. Left panel: the solid line is the predicted in-sample and out-of-sample accumulation curve computed by averaging over posterior samples of $E(K_n)$ and $E(K_{n+m} \mid D_1, \dots , D_n)$, respectively. Black dashed lines indicate the 95\% posterior predictive credible intervals. Right panel: the solid line is the predicted accumulation curve obtained averaging posterior samples of $E(K_n)$}
    \label{fig:LL3_models}
\end{figure}

\begin{figure}
\centering
\includegraphics[width=\linewidth]{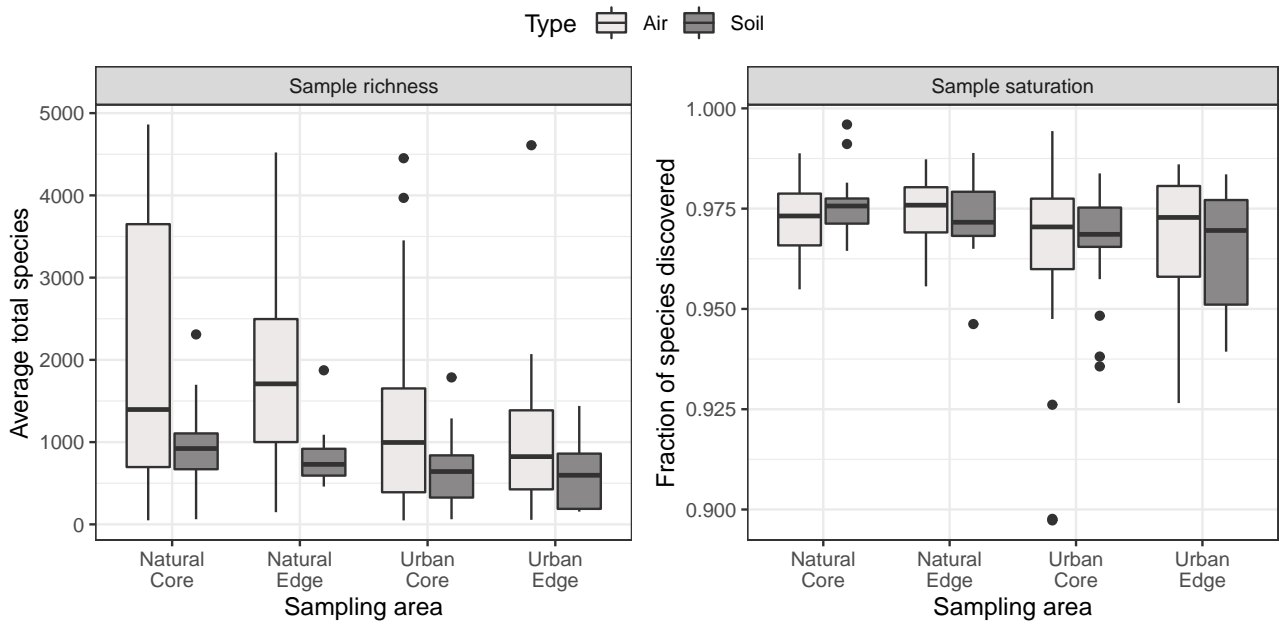}
    \caption{Left panel: distribution of the posterior mean sample species richness for the 166 samples. Right panel: distribution of the posterior mean sample saturation  for the 166 samples.} 
    \label{fig:Richness_saturations}
\end{figure}


We rely on model \textsc{ll}-3 in performing inferences on i) the sample species richness, which is the total number of species that can be detected through barcoding within a sample, and ii) whether \textsc{dna} barcoding has reached {\em saturation} at different sites, meaning that only very few species are missed.
To address i), we estimate the posterior mean $E(K_\infty\mid D_1, \ldots, D_n)$ for each individual sample, which is guaranteed to be finite. The results are reported in the left panel of Figure~\ref{fig:Richness_saturations}, which displays the expected sample species richness for each of the 166 samples across
site characteristics.  Air samples tend to contain more species, and there is some evidence of greater species richness in natural environments, as reported by \citet{Abrego_2020}.

For task ii), let $C_n = K_n/K_\infty \le 1$  represent the saturation level of a given sample after $n$ barcoded sequences. 
Differences across sites can be evaluated via $E(C_n \mid D_1, \ldots, D_n)$, which represents the posterior expected saturation level of a sample. 
Figure~\ref{fig:Richness_saturations}, right panel, summarizes posterior mean saturation stratified by sampling site characteristics.
 While there is some variability across sites, most of them have a ratio greater than 
 $0.95$. The results suggest that if additional
  \textsc{dna} sequences are barcoded there is the opportunity to detect approximately $2.5$-$5$\% more species in each sample. 
 We can estimate the number $m$ of additional sequences that would need to be barcoded to reach a desired saturation level 
 $C_{n+m}$.  For example, for the sample highlighted in Figures~\ref{fig:dir_failure} and~\ref{fig:LL3_models}, the saturation level is $98.0$\%. To achieve a saturation level of $99.5$\% we would need to barcode an additional $41,547$ sequences, an increase of $43\%$.




\section*{Acknowledgements}
This project has received funding from the European Research Council under the European Union’s Horizon 2020 research and innovation programme (grant agreement No 856506).


\appendix
\section{Proofs}

\begin{proof}[Proof of Theorem~\ref{theo:EPPF_Bernoulli}]
 We first discuss the likelihood $\mathscr{L}(\alpha\mid X_1, \ldots, X_n)$. Given the sequence of tags $X_1, \ldots, X_n$, any exchangeable prediction scheme defines a random partition $\Psi_n$ of the integers $\{1,\dots,n\}$ such that $i$ and $j$ belong to the same set in $\Psi_n$ if and only if $X_i = X_j$. 
Let $\{C_1,\dots,C_k\}$ be a partition of $n$ into $k$ groups, with $n_j = \text{card}(C_j)$, $(j=1,\dots,k)$ being the cardinality of $C_j$. The resulting law of the random partition $\Psi_n$ 
in the Dirichlet process case equals $\textrm{pr}(\Psi_n = \{C_1,\dots,C_k\}) = \alpha^k/(\alpha)_{n}\prod_{j=1}^k(n_{j}-1)!$.
This is the likelihood function for  $\alpha$, so that $\mathscr{L}(\alpha \mid X_1,\dots,X_n) \propto \alpha^k/(\alpha)_n$, which only depends on $K_n = k$ and not on $n_1,\dots,n_k$. By letting $(\alpha)_n =  \alpha\prod_{i=2}^n(\alpha + i - 1)$, we get
\begin{equation}\label{eq:logEPPF}
    \log \mathscr{L}(\alpha\mid X_1, \ldots, X_n) = (k-1)\log \alpha - \sum_{i=2}^n \log(\alpha + i -1) + c_X,
\end{equation}
where $c_X$ is a constant not depending on $\alpha$. On the other hand, the logarithm of the likelihood induced by the discovery indicators is equal to
\begin{equation}\label{eq:logLik_ind}
  \log\mathscr{L}(\alpha \mid D_1, \ldots, D_n) = \log\alpha\sum_{i=2}^n D_i - \sum_{i=2}^n \log(\alpha + i -1) + c_D,
\end{equation}
with $c_D$ being a constant not depending on $\alpha$. Since $\sum_{i=2}^n D_i = k-1$, one has that equation~\eqref{eq:logEPPF} and \eqref{eq:logLik_ind} are equal up to an additive constant. Thus, the result follows.
\end{proof}

\begin{proof}[of Proposition~\ref{pro:PoisPrediction}]
This proof follows from the definition of a Poisson-binomial distribution. In particular, for every $n, m \geq 1$, we have that $K_{m}^{(n)} = K_{n+m} - K_{n} = \sum_{i = n+1}^{n+m}D_i$ is a sum of independent indicators with discovery probabilities $\pi_{n+1} \ldots \pi_{n+m}$ directed by $S(t;\theta)$. Hence, from Definition~\ref{def1}, $(K_m^{(n)}\mid D_1, \ldots, D_n) \sim \textsc{pb}\{S(n;\theta), \ldots, S(n+m-1;\theta)\}$, whose expected value is $E(K_m^{(n)}\mid D_1, \ldots, D_n) = \sum_{j=1}^m S(j+n-1;\theta)$. Finally, if $K_n =k$, it naturally follows that $E(K_{n+m}\mid D_1, \ldots, D_n) = E(K_n\mid D_1, \ldots, D_n) + E(K_{m}^{(n)}\mid D_1, \ldots, D_n) = k + E(K_{m}^{(n)}\mid D_1, \ldots, D_n).$
\end{proof}

\begin{proof}[of Proposition~\ref{pro:ET}]
Recall that $K_n = \sum_{i=1}^n D_i$ is non-decreasing in $n$. Taking the limit as $n\to \infty$, we have that $K_n \to K_\infty =\sum_{i=1}^{\infty} D_i$ almost surely. Then,  $E(K_\infty) = \sum_{i=1}^\infty S(i-1; \theta)$, as a consequence of the monotone convergence theorem. Moreover, equation~\eqref{eq:ET_ineq} follows from a simple calculus inequality: as $S(t;\theta)$ is positive and strictly decreasing in $t$, we have that
$$
\int_{0}^{n-1} S(t;\theta)\mathrm{d}t + S(n-1;\theta) \leq \sum_{i=1}^{n}  S(i-1;\theta) \leq \int_{0}^{n-1}S(t;\theta)\mathrm{d}t + S(0;\theta),
$$ for every $n \ge 1$. Taking the limit for $n\to\infty$, one has that  $S(n-1;\theta) \to 0$. But then, as $S(0;\theta)=1$, we have that
$E(T) \leq E(K_\infty) \leq E(T) +1$, with $E(K_\infty)$ defined above and $E(T) =\int_{0}^\infty S(t;\theta) \mathrm{d}t$.
\end{proof}

\begin{proof}[Proof of Corollary~\ref{cor:ET_inequalities}]
We begin by proving that $K_\infty = \infty $ if and only if $E(K_\infty) =\infty$. One side follows from the monotone convergence theorem: if $K_\infty = \infty$, then necessarily $\lim_{n\to \infty}E(K_n) = E(K_\infty)= \infty$ by the same argument in the proof of Proposition~\ref{pro:ET}. The other direction can be proved by contrapposition: suppose that  $K_\infty < \infty$. Then, there exists a positive constant $M < \infty$ such that $K_\infty < M$ almost surely. This means that $ E(K_\infty) < E(M) = M<\infty$. The rest of the claim naturally follows from the inequality in equation~\eqref{eq:ET_ineq} of Proposition~\ref{pro:ET}. 
\end{proof}

\begin{proof}[of Corollary~\ref{cor:ET_var}]
To prove this claim, we rely on the limit comparison test for the ratio of two series. In particular, 
$$
\lim_{n\to \infty} \frac{ S(n-1;\theta)\{1-S(n-1;\theta)\}}{S(n-1;\theta)} = 1 - \lim_{n \to \infty} S(n-1; \theta) = 1. 
$$
This implies that $\textup{var}(K_\infty) =  \sum_{i=1}^{\infty}S(i-1;\theta)\{1-S(i-1;\theta)\}$ diverges if and only if $E(K_\infty) = \sum_{i=1}^{\infty}S(i-1;\theta)$ diverges. Following the same argument in the proof of Corollary~\ref{cor:ET_inequalities}, having $K_\infty<\infty$ almost surely implies that $E(K_\infty)< \infty$, and in turn $\textup{var}(K_\infty) <\infty$. 
\end{proof}


\begin{proof}[Proof of Theorem~\ref{theo:CLT}]
The first part of the theorem is a consequence of the strong law of large numbers for the sum of independent random variables. In particular, let $b_n = \int_{1}^{n} S(t-1;\theta)\mathrm{d}t$. Then, $b_n < b_{n+1}$ for every $n$, and $b_n \to E(T) = \infty$ as $n\to \infty$. Since by assumption $S(n-1; \theta) > S(n; \theta)$ and $b_n^2 < b_{n+1}^2$ for every $n$, we have that $\sum_{n=1}^\infty \textup{var}(D_n)/b_n^2 < \infty$, which holds by the series convergence test, because
\begin{align*}
    \lim_{n\to \infty} \frac{\textup{var}(D_{n+1})}{b_{n+1}^2} \frac{b_{n}^2}{\textup{var}(D_{n})} &= \lim_{n\to\infty} \frac{S(n;\theta)\{1-S(n;\theta)\}}{S(n-1;\theta)\{1-S(n-1;\theta)\}}\frac{b_n^2}{b_{n+1}^2} < \lim_{n\to\infty} \frac{1-S(n;\theta)}{1-S(n-1;\theta)}= 1.
\end{align*}
Hence, the above condition ensures that $\{K_n-E(K_n)\}/b_n \to 0$ almost surely as $n\to \infty$ by the strong law of large numbers. This means that $\lim_{n\to \infty}K_n/b_n= \lim_{n\to \infty} E(K_n)/b_n = 1$,
almost surely, as a consequence of Proposition~\ref{pro:ET}.

The second part of the claim follows from Lyapunov's central limit theorem. Define $\sigma^2_n = \textup{var}(K_n)$ for every $n$. As the discovery indicators $(D_n)_{n\geq 1}$ are all independent, we can prove the central limit theorem for $K_n$ by showing that there exists a $\delta > 0$ such that $\lim_{n\to \infty}1/\sigma_n^{2+\delta}\sum_{i = 1}^n E(|D_i - \pi_i|^{2+\delta}) = 0$, where $\pi_n = S(n-1;\theta)$ is the discovery probability at every $n$. Fix $\delta = 2$. From the proofs of Corollaries~\ref{cor:ET_inequalities} and~\ref{cor:ET_var}, we have that $K_\infty=\infty$ implies that $ \lim_{n\to\infty} \sigma^2_n =\infty$. Moreover, by looking at the fourth centered moment of a Bernoulli distribution we have that
$$
\sum_{i=1}^n E(|D_i - \pi_i|^4) =\sum_{i=1}^n  \pi_i(1-\pi_i)\{1- 3\pi_i(1-\pi_i)\} \leq \sum_{i=1}^n  \pi_i(1-\pi_i)= \sigma^2_n,$$
which leads to $0 \leq \lim_{n\to \infty} 1/\sigma_n^{4} \sum_{i = 1}^n E(|D_i - \pi_i|^{4}) \leq \lim_{n\to\infty} 1/\sigma_n^2 = 0$, concluding the proof.
\end{proof}

\begin{proof}[Proof of Proposition~\ref{prop:convergence}]
This can be proved by means of the series convergence test. By the fact that
$$
\lim_{n\to\infty} \frac{S(n;\alpha, \sigma, \phi)}{S(n-1;\alpha, \sigma, \phi)} = \lim_{n\to\infty} \frac{\alpha\phi^{n}}{\alpha\phi^{n} + n^{1-\sigma}} \frac{\alpha\phi^{n-1} + (n-1)^{1-\sigma}}{\alpha\phi^{n-1}} = \phi,
$$
having $\phi<1$ implies that 
$E(K_\infty) = \sum_{i=1}^\infty S(i-1;\alpha, \sigma, \phi) < \infty$ almost surely. But then, $K_\infty < \infty$ as well by the proof of Corollary~\ref{cor:ET_inequalities}.
\end{proof}

\begin{proof}[Proof of Theorems~\ref{teo:C_recursion} and~\ref{theo:Kn_pmf}]
The proofs of Theorems~\ref{teo:C_recursion} and~\ref{theo:Kn_pmf} are presented together. The arguments we use follow a similar line of reasoning as in \citet{Charalambides_2005}.
As a first step, we prove the triangular recurrence in Theorem~\ref{teo:C_recursion}. Following Definition~\ref{def2}, we can write $\prod_{k=0}^{n}(\alpha + k^{1-\sigma}\phi^{-k}) = (\alpha + n^{1-\sigma}\phi^{-n})\prod_{k=0}^{n-1}(\alpha + k^{1-\sigma}\phi^{-k}),$
for any $n\geq 1$, from which it follows that
\begin{align*}
    \sum_{k=0}^{n+1} \alpha^k \mathscr{C}_{n+1,k}(\sigma, \phi) =  \sum_{k=1}^{n+1} \alpha^k \mathscr{C}_{n,k-1}(\sigma, \phi)  + \sum_{k=0}^n \alpha^k n^{1-\sigma}\phi^{-n} \mathscr{C}_{n,k}(\sigma, \phi).
\end{align*}
Hence, all the coefficients associated to each $\alpha^k$ must coincide under both sides of the above equation. This means that $\mathscr{C}_{n+1,k}(\sigma, \phi) = \mathscr{C}_{n,k-1}(\sigma, \phi) +  n^{1-\sigma}\phi^{-n}\mathscr{C}_{n,k}(\sigma, \phi)$. As for the initial conditions, it is easy to check that they naturally follow from Definition~\ref{def2}.

To prove the second part of Theorem~\ref{teo:C_recursion}, we start by considering $\mathrm{pr}(K_n=k)$. Call $j_1,\ldots, j_n$ a sequence of indexes such that $D_{j_s}=1$ for $s = 1, \ldots, k$, and $D_{j_s}=0$ for $s = k+1, \ldots, n$. By independence of the indicators, the probability of such a configuration is  
\begin{align*}
&\mathrm{pr}(D_{j_1}=0,\ldots, D_{j_k} = 1, D_{j_{k+1}} =0, \ldots, D_{j_n} = 0) =\\
&\quad=\prod_{s = 1}^k S(j_s - 1; \alpha, \sigma, \phi) \prod_{s=k+1}^n \{1-S(j_s - 1; \alpha, \sigma, \phi)\}
= \frac{\alpha^k}{\prod_{i=0}^{n-1} (\alpha + i^{1-\sigma}\phi^{-i})}\prod_{j=1}^{n-k} i_j^{1-\sigma}\phi^{-i_j},
\end{align*}
where the product in the last equality follows from  relabeling the indexes as $i_1 = j_{k+1}-1, \ldots, i_{n-k} = j_{n}-1$. Moreover, note that $\{i_1, \ldots, i_{n-k}\}$ is one of the $n-k$ possible combinations of the $n-1$ positive integers $\{1, \ldots, n-1\}$ for which we obtain precisely $k$ discoveries, with $1\leq k \leq n$ and $n\geq 2$. Hence, summing over all the possible combinations of $\{i_1, \ldots, i_{n-k}\}$ leads us to the probability 
\begin{equation}\label{eq:prooftheo34_step1}
\mathrm{pr}(K_n = k) =  \frac{\alpha^k}{\prod_{i=0}^{n-1} (\alpha + i^{1-\sigma}\phi^{-i})}\sum_{(i_1, \ldots, i_{n-k})}\prod_{j=1}^{n-k} i_j^{1-\sigma}\phi^{-i_j},
\end{equation}
for $1\leq k \leq n$ and $n\geq 2$.
The object in equation~\eqref{eq:prooftheo34_step1} is a probability mass function. This means that
$$
\sum_{k=0}^n \mathrm{pr}(K_n = k) = \sum_{k=0}^n \frac{\alpha^k}{\prod_{i=0}^{n-1} (\alpha + i^{1-\sigma}\phi^{-i})}\sum_{(i_1, \ldots, i_{n-k})}\prod_{j=1}^{n-k} i_j^{1-\sigma}\phi^{-i_j} = 1,
$$
recalling that $\mathrm{pr}(K_n = 0) = 0$. Rearranging the equality, one has that 
$$
\prod_{i=0}^{n-1} (\alpha + i^{1-\sigma}\phi^{-i}) = \sum_{k=0}^n \alpha^k \sum_{(i_1, \ldots, i_{n-k})}\prod_{j=1}^{n-k} i_j^{1-\sigma}\phi^{-i_j},
$$
which is the same polynomial expansion proposed in Definition~\ref{def2}. Hence, it must be that 
\begin{equation}\label{eq:prooftheo34_step2}
    \mathscr{C}_{n, k}(\sigma, \phi) = \sum_{(i_1, \ldots, i_{n-k})}\prod_{j=1}^{n-k}i_j^{1-\sigma}\phi^{-i_j},
\end{equation}
again for $1\leq k \leq n$ and $n\geq 2$.
This last equality proves the second part of Theorem~\ref{teo:C_recursion}. Finally, Theorem~\ref{theo:Kn_pmf} naturally follows by plugging equation~\eqref{eq:prooftheo34_step2} into \eqref{eq:prooftheo34_step1}.
\end{proof}

\begin{proof}[Proof of Proposition~\ref{pro:regression}]
To find the maximizer $\hat{\theta} = (\hat{\alpha}, \hat{\sigma}, \hat{\phi})$ of the likelihood in equation~\eqref{eq:LikBern} with the the three-parameter log-logistic specification we rely on the first order condition with respect to $\alpha$.  In particular, the logarithm of the likelihood becomes 
$$
\log \mathscr{L}(\alpha,\sigma,\phi\mid D_1, \ldots D_n) = \log\alpha \sum_{i=2}^n D_i  - \sum_{i=2}^n \log\{\alpha\phi^{i-1} + (i-1)^{1-\sigma}\} + c_{\sigma, \phi},
$$
where $c_{\sigma,\phi}$ is a constant not dependent on $\alpha$.  Hence, the first order condition with respect to $\alpha$ leads to
\begin{equation*}
\sum_{i=1}^n \frac{\alpha\phi^{i-1}}{\alpha\phi^{i-1} + (i-1)^{1-\sigma}} = \sum_{i=1}^n D_i = k = E(K_n).
\end{equation*}
This equality must be maintained at the solution $\hat{\theta}$.

\end{proof}

\section{Supplementary material}
\subsection{Posterior inference for a single accumulation curve}
In the following we describe the estimation of the parameters $\theta = (\alpha,\sigma,\phi)$ under the three-parameter log-logistic specification and using Markov Chain Monte Carlo. 
Let $(D_n)_{n\geq 1}$ be a sequence of discovery indicators with $D_1=1$ and
$$\pi_{n+1} = \mathrm{pr}(D_{n+1} = 1\mid D_1, \ldots, D_n) = \frac{\alpha\phi^{n}}{\alpha\phi^{n} + n^{1-\sigma}}, \qquad n \geq 1,$$
for $\alpha> 0$, $\sigma < 1$, $0 < \phi \le 1$ and  $\pi_1 = 1$. As discussed in the manuscript, this implies that
\begin{equation*}
\log \frac{\pi_{n+1}}{1-\pi_{n+1}} = \log \alpha - (1-\sigma)\log n + (\log \phi)n = \beta_0 + \beta_1\log{n} + \beta_2 n, \qquad n\geq 1,
\end{equation*}
with  $\beta_0 = \log \alpha$, $\beta_1 = \sigma - 1 < 0$ and $\beta_3 = \log \phi \le  0$. These constraints are imposed through a truncated normal prior, namely
$\beta \sim N(\mu, \Sigma)\mathds{1}(\beta_1 <0; \beta_2 \leq 0)$.

Samples from the posterior can be easily obtained via the P\'olya-gamma data-augmentation strategy introduced in \citet{Polson_2012}. This procedure introduces 
Pólya-gamma distributed positive latent variables $\omega = (\omega_2, \ldots, \omega_n)^\mathrm{T}$. The resulting full conditional  distributions for $\beta$ and $\omega$ 
are available in closed form. Let $d = (d_2, \ldots, d_n)^\mathrm{T}$ be the observed values for the discovery indicators $D_2, \ldots, D_n$ and let $V$ be the  design matrix, with $n-1$ rows and $3$ columns and entries $v_i = (1, \log{i}, i)^\mathrm{T}$, for $i = 1,\ldots, n-1$. Then, the full conditional for $(\beta\mid \omega, d)$ is a multivariate truncated normal distribution with parameters $\mu_\omega$ and $\Sigma_\omega$ equal to
\begin{equation}\label{eq:full_conds}
    \Sigma_\omega = (V^{\mathrm{T}}\Omega  V+ \Sigma^{-1}\mu)^{-1}, \quad \mu_\omega =\Sigma_\omega(V^\mathrm{T}\kappa + \Sigma^{-1}\mu),
\end{equation}
with $\kappa = (d_2 - 1/2, \ldots, d_n - 1/2)^\mathrm{T}$ and $\Omega = \mathrm{diag}(\omega_2, \ldots, \omega_n)$. The algorithm below outlines the sampling procedure.

\begin{algorithm}
\caption{P\'olya-Gamma Gibbs sampler for single site accumulation curve}\label{algo:logistic_simple}
\begin{algorithmic}[1]
\State Set an initial value $\beta$ and set number of samples $R$
\For{$r=1$ to $r = R$}
\For{$i=1$ to $i=n-1$}
\State Sample $(\omega_i \mid \beta) \sim \textrm{PolyaGamma}(1, v_{ i}^\mathrm{T}\beta)$
\EndFor
\State Sample $(\beta \mid \omega, d) \sim N(\mu_\omega, \Sigma_\omega)\mathds{1}({\beta_1 <0, \beta_2 \leq 0})$, with $\mu_\omega$, $\Sigma_\omega$ as in~\eqref{eq:full_conds}
\EndFor
\State \small{\textbf{Output}}: collection of $R$ samples for $\beta$
\end{algorithmic}
\end{algorithm}


In our work we obtain samples from the multivariate truncated normal through the efficient algorithm proposed in \cite{Botev_2017}. 

\subsection{Posterior sampling for multi-site data}

We now describe a Markov Chain Monte Carlo algorithm for Bayesian inference for the covariate-dependent model described in the manuscript. Recall that we are given a collection of $L$ accumulation curves $(K_{1 n})_{n \ge 1}, \dots,(K_{L n})_{n \ge 1}$ observed up to the terms $n_1,\dots,n_L$. Each curve is associated to a set of covariates $z_\ell= (z_{\ell 1}, \ldots, z_{\ell p})^\mathrm{T}$ for $\ell =1,\dots,L$. Future observations correspond to new discoveries within the set of the considered $L$ curves, so that new covariates values are not expected. 

Let $(D_{\ell n})_{n \ge 1}$ be the sequence of discovery indicators for the $\ell$th location, with probabilities $(\pi_{\ell n})_{n \ge 1}$. Hence, we get
\begin{equation*}\label{eq:regress_covariate}
\log \frac{\pi_{\ell n+1}}{1-\pi_{\ell n+1}} =  \beta_{\ell 0} + \beta_{\ell 1}\log{n} + \beta_{\ell 2} n = z_\ell^\mathrm{T} \gamma_0 + (z_\ell^\mathrm{T} \gamma_1) \log{n}  + (z_\ell^\mathrm{T} \gamma_2) n,
\end{equation*}
with $\gamma_0, \gamma_1, \gamma_2 \in \mathds{R}^p$ coefficient vectors such that $z_\ell^\mathrm{T} \gamma_2 < 0$ and $z_\ell^\mathrm{T} \gamma_2 \le 0$ for every $\ell = 1,\dots,L$. The above specification is a logistic regression and  therefore inference on the parameters $\gamma = (\gamma_0,\gamma_1, \gamma_2)^{\mathrm{T}}$ may be conducted through a simple modification of Algorithm~\ref{algo:logistic_simple}.

Let $N = \sum_{\ell=1}^L (n_\ell-1)$ and let $V$ be a design matrix with $N$ rows and $3 p$ columns, with rows $v_{(\ell i)} = (z^{\mathrm{T}}_\ell, z^{\mathrm{T}}_\ell \log{i}, z^{\mathrm{T}}_\ell i)^{\mathrm{T}}$ for $i= 1, \ldots, n_\ell-1$ and $\ell=1,\dots,L$. Moreover, call $d = (d_{1}^\mathrm{T}, \ldots, d_L^\mathrm{T})^\mathrm{T}$ the realized discoveries, with $d_\ell = (d_{\ell 2}, \ldots, d_{\ell  n_\ell})^\mathrm{T}$ be the observed values for  $D_{\ell1}, \ldots, D_{\ell n}$, for every $\ell = 1, \ldots, L$. As before, we can incorporate the constraints  $z^{\mathrm{T}}_\ell\gamma_1<0$ and $z^{\mathrm{T}}_\ell\gamma_2\le 0$ by assigning $\gamma$ a multivariate truncated normal prior, 
$$\gamma \sim N(\mu, \Sigma)\mathds{1}( z^{\mathrm{T}}_\ell\gamma_1 <0 ; z^{\mathrm{T}}_\ell\gamma_2 \le 0; \ell = 1,\ldots, L).$$
Let $\omega$ be a $N$-dimensional vector of P\'olya-gamma latent  variables. Then, the full conditional for $(\gamma\mid \omega, d)$ is a multivariate truncated normal distribution  with mean $\mu_\omega$ and covariance matrix $\Sigma_\omega$ equal to equation~\eqref{eq:full_conds}, whereas $\Omega$ is a diagonal matrix whose diagonal elements are those of the vector $\omega$. 

In Algorithm~\ref{algo:logistic_multi} we employ a vanilla acceptance rejection sampler for the full conditional $(\gamma \mid \omega, d)$. This is indeed a reasonable approach in most practical settings, as the data usually support the required constraints, leading to very high acceptance rates. If needed, suitable adaptations of the ideas of \citet{Botev_2017} may be alternatively considered. 

\begin{algorithm}
\caption{P\'olya-Gamma Gibbs sampler for covariate-dependent accumulation curves}\label{algo:logistic_multi}
\begin{algorithmic}[1]
\State Set an initial value $\gamma$ and set the number of samples $R$ 
\For{$r=1$ to $r = R$}
\For{$\ell=1$ to $\ell=L$}
\State Sample $(\omega_{(\ell i)} \mid \gamma) \sim \textrm{PolyaGamma}(1, v_{(\ell i)}^\mathrm{T}\gamma), \quad i = 1, \ldots, n_\ell-1$
\EndFor
\State Sample $(\gamma \mid \omega, d) \sim N(\mu_\omega, \Sigma_\omega)$, with $\mu_\omega$, $\Sigma_\omega$ as in~\eqref{eq:full_conds} until $\gamma$\\ \qquad \qquad \qquad satisfies $z^{\mathrm{T}}_\ell\gamma_1 <0$ and $z^{\mathrm{T}}_\ell\gamma_2 \le 0$ for every $\ell = 1,\ldots, L$. 
\EndFor
\State \small{\textbf{Output}}: collection of $R$ samples for $\gamma$
\end{algorithmic}
\end{algorithm}

   

\bibliographystyle{chicago}
\bibliography{paper-ref}

\end{document}